\documentclass[a4paper,oneside]{scrartcl}

\usepackage{enumerate}
\usepackage[english]{babel}
\usepackage{amsmath}
\usepackage{amsthm}
\usepackage{amssymb}
\usepackage{amsfonts}
\usepackage{ifthen}
\usepackage{xspace}
\usepackage{xcolor}
\usepackage{graphicx}
\usepackage{fancyhdr}

\newcommand{\logicClFont}[1]{\mathsf{#1}}        
\newcommand{\numberClassFont}[1]{\mathbb{#1}}     
\newcommand{\complexityClassFont}[1]{\textsl{#1}} 
\newcommand{\mathCommandFont}[1]{\mathrm{#1}}     
\newcommand{\problemFont}[1]{\textsc{#1}}     

\newcommand{\literal}[1]{{\protect\ensuremath{#1}}\xspace}
\newcommand{\problem}[1]{\literal{\problemFont{#1}}}
\newcommand{\logic}[1]{\literal{\logicClFont{#1}}}
\newcommand{\commandOperator}[2]{\literal{\mathord{\mathCommandFont{#1}\ifthenelse{\equal{#2}{}}{}{(\nobreak#2\nobreak)}}}}

\newcommand{\cmd}[1]{\commandOperator{#1}{}}

\newcommand{\todo}[2][]{\textnormal{\color{red}\scriptsize+++\ifthenelse{\equal{#1}{}}{}{#1 says: }#2+++}}

\newcommand{\N}{\protect\ensuremath{\numberClassFont{N}}}

\newcommand{\NEXPTIME}{\literal{\complexityClassFont{NEXPTIME}}}
\newcommand{\DNEXPTIME}{\literal{\complexityClassFont{2NEXPTIME}}}
\newcommand{\sigmazeroone}{\literal{\mathit{\Sigma}^0_1}}
\newcommand{\sigmaoneone}{\literal{\mathit{\Sigma}^1_1}}
\newcommand{\pizeroone}{\literal{\mathit{\Pi}^0_1}}
\newcommand{\true}{\protect\raisebox{-1pt}{$\true$}}
\newcommand{\false}{\protect\ensuremath\raisebox{-1pt}{$\false$}}
\newcommand{\dep}[1][]{\commandOperator{=}{#1}}
\newcommand{\fr}[1][]{\commandOperator{Fr}{#1}}
\newcommand{\ar}[1][]{\commandOperator{arity}{#1}}
\newcommand{\var}[1][]{\commandOperator{Var}{#1}}
\newcommand{\dom}[1][]{\commandOperator{dom}{#1}}
\newcommand{\rng}[1][]{\commandOperator{range}{#1}}
\newcommand{\rel}[1][]{\commandOperator{rel}{#1}}
\newcommand{\df}{\logic{D}}
\newcommand{\ifl}{\logic{IF}}
\newcommand{\dtwo}{\literal{\df^2}}
\newcommand{\iftwo}{\literal{\ifl^2}}
\newcommand{\eso}{\logic{ESO}}
\newcommand{\fo}{\logic{FO}}
\newcommand{\fotwo}{\literal{\fo^2}}
\newcommand{\foc}{\logic{FOC}}
\newcommand{\foctwo}{\logic{FOC^2}}
\newcommand{\mA}{\literal{\mathfrak{A}}}
\newcommand{\mB}{\literal{\mathfrak{B}}}
\newcommand{\mG}{\literal{\mathfrak{G}}}
\newcommand{\mH}{\literal{\mathfrak{H}}}
\newcommand{\mD}{\literal{\mathfrak{D}}}
\newcommand{\mM}{\literal{\mathfrak{M}}}
\renewcommand{\true}{\literal{\top}}
\renewcommand{\false}{\literal{\bot}}
\newcommand{\sat}[1][]{\commandOperator{\problem{Sat}}{#1}}
\newcommand{\finsat}[1][]{\commandOperator{\problem{FinSat}}{#1}}
\newcommand{\calL}{\literal{\mathcal{L}}}
\newcommand{\tiling}[1][]{\commandOperator{\problem{Tiling}}{#1}}
\newcommand{\torus}{\literal{\mathrm{Torus}}}
\newcommand{\I}{\literal{\mathrm{I}}}

\theoremstyle{definition}
\newtheorem{definition}{Definition}[section]
\theoremstyle{plain}
\newtheorem{theorem}[definition]{Theorem}
\newtheorem{lemma}[definition]{Lemma}

\newtheorem{proposition}[definition]{Proposition}
\newtheorem{remark}[definition]{Remark}

\newcommand{\Title}{Complexity of two-variable Dependence Logic and IF-Logic}
\newcommand{\Author}{
Juha Kontinen%
\thanks{University of Helsinki, juha.kontinen@helsinki.fi}
, Antti Kuusisto%
\thanks{University of Tampere, \{antti.j.kuusisto, jonni.virtema\}@uta.fi}
, Peter Lohmann%
\thanks{Leibniz University Hannover, lohmann@thi.uni-hannover.de}
, Jonni Virtema%
\footnotemark[3]
}
\newcommand{\PDFAuthor}{Juha Kontinen, Antti Kuusisto, Peter Lohmann, Jonni Virtema}
\newcommand{\Keywords}{dependence logic, independence-friendly logic, two-variable logic, decidability, complexity, satisfiability, expressivity}
\newcommand{\SubjClass}{F.4.1 Computability theory, Model theory; F.1.3 Reducibility and completeness}

\usepackage[
  colorlinks=false,pdfdisplaydoctitle=false,pdfkeywords={\Keywords},pdftitle={\Title},pdfauthor={\PDFAuthor}
]{hyperref}

\begin{document}

\title{\Title\thanks{This work was supported by grants 127661, 129208, 129761, 129892 and 138163 of the Academy of Finland and by DAAD grant 50740539}}
\author{\Author}

\maketitle

\begin{abstract}
\small
We study the two-variable fragments $\dtwo$ and $\iftwo$ of dependence logic and in\-de\-pen\-dence-friendly logic. We consider the satisfiability and finite satisfiability  problems of these logics and show that for $\dtwo$, both problems are \NEXPTIME-complete, whereas for $\iftwo$, the problems are \pizeroone and \sigmazeroone-complete, respectively. We also show that $\dtwo$ is strictly less expressive than $\iftwo$ and that already in $\dtwo$, equicardinality of two unary predicates and infinity can be expressed (the latter in the presence of a constant symbol).

This is an extended version of a publication in the proceedings of the 26th Annual IEEE Symposium on Logic in Computer Science (LICS 2011).
\end{abstract}


\smallskip
{\bfseries ACM Subject Classifiers:}
\SubjClass

\thispagestyle{fancy}
\renewcommand{\headrulewidth}{0pt}
\renewcommand{\footrulewidth}{0pt}
\fancyhf{}
\fancyfoot[L]{\scriptsize \textcopyright\,2011 IEEE (for the original version of this article).\\
Personal use of this material is permitted. Permission from IEEE must be obtained for all other users, including reprinting/ republishing this material for advertising or promotional purposes, creating new collective works for resale or redistribution to servers or lists, or reuse of any copyrighted components of this work in other works.\\
The original publication is (or will soon be) available at ieeexplore.ieee.org}

\section{Introduction}

The satisfiability problem of first-order logic \fo was shown to be undecidable in \cite{ch36,tu36}, and ever since, logicians have been searching for decidable fragments of \fo.
Henkin \cite{he67} was the first to consider the logics $\fo^k$, i.e., the fragments of first-order logic with $k$ variables.
The fragments $\fo^k$, for $k\ge 3$, were easily seen to be undecidable but the case for $k=2$ remained open.
Scott \cite{sc62} then showed that \fotwo without equality is decidable.  Mortimer \cite{mo75} extended the result to \fotwo with equality and showed that every satisfiable \fotwo formula has a model whose size is doubly exponential in the length of the formula. His result established that the satisfiability and finite satisfiability problems of \fotwo are contained in \DNEXPTIME.
Finally, Gr\"adel, Kolaitis and Vardi \cite{grkova97} improved the result of Mortimer by establishing that every satisfiable \fotwo formula  has a model of  exponential size. Furthermore, they showed that the satisfiability problem for \fotwo is \NEXPTIME-complete.

The decidability of the satisfiability problem of various extensions of \fotwo has been studied (e.g.~\cite{grotro97, grot99, etvawi02, kiot05}). One such interesting extension \foctwo is acquired by extending \fotwo with counting quantifiers $\exists^{\geq i}$. The meaning of a formula of the form $\exists^{\geq i} x \phi(x)$ is that $\phi(x)$ is satisfied by  at least $i$ distinct elements.
The satisfiability problem for the logic \foctwo
was shown to be decidable by Gr\"adel et al.~\cite{grotro97a}, and shown to be in \DNEXPTIME by Pacholski et al.~\cite{paszte97}. 
Finally, Pratt-Hartmann \cite{Pratt-Hartmann:2005} established that the problem is \NEXPTIME-complete.
We will later use the result of Pratt-Hartmann to determine the complexity of the satisfiability problem of the two-variable fragment of dependence logic.

\begin{table}[!t]\label{table:results}
\begin{center}
\begin{tabular}{ccc}
\emph{Logic}            & \emph{Complexity of \sat~/ \finsat}   & \multicolumn{1}{c}{\emph{References}}\\
\fo, $\fo^3$            & \pizeroone ~/ \sigmazeroone           & \cite{ch36,tu36}\\
\eso, \df, \ifl         & \pizeroone ~/ \sigmazeroone           & Remark~\ref{eso sat}, \cite{ch36,tu36}\\
\fotwo                  & \NEXPTIME                             & \cite{grkova97}\\
\foctwo                 & \NEXPTIME                             & \cite{Pratt-Hartmann:2005}\\
\dtwo                   & \NEXPTIME                             & Theorem~\ref{dtwo nexptime}\\
$\fotwo(\I)$        & \sigmaoneone-hard / \sigmazeroone     & \cite{grotro97}\\
\iftwo                  & \pizeroone ~/ \sigmazeroone           & Theorems~\ref{iftwo sat complexity},~\ref{iftwo finsat complexity}\\
\end{tabular}
\caption{Complexity of satisfiability for various logics.\newline
The results are completeness results for the full relational vocabulary.}
\end{center}
\end{table}

In this article we study the satisfiability of the two-variable fragments of independence-friendly logic (\ifl) and dependence logic (\df). The logics $\ifl$ and $\df$ are conservative extensions of $\fo$, i.e., they agree with $\fo$ on sentences which syntactically are $\fo$-sentences. We thereby contribute to the understanding of the satisfiability problems of extensions of $\fotwo$. We briefly recall the history of $\ifl$ and $\df$. In first-order logic the order in which quantifiers are written determines dependence relations between variables. For example, when using game 
theoretic semantics to evaluate the formula
\[ 
\forall x_0\exists x_1\forall x_2\exists x_3\, \phi,
\] 
 the choice for $x_1$ depends on the value for $x_0$, and the choice for $x_3$   depends on the value of both universally quantified variables $x_0$  and  $x_2$. The characteristic feature of \df and \ifl is that in these logics it is possible to express dependencies between variables  that cannot be expressed in \fo.
The first step in this direction was taken by Henkin \cite{henkin1961} with his partially ordered quantifiers
\begin{equation}\label{poc}
\left(\begin{array}{cc}\forall x_0& \exists x_1\\ \forall
x_2&\exists x_3\end{array}\right)\phi,
\end{equation}
where $x_1$ depends only
on $x_0$ and  $x_3$ depends only on $x_2$. Enderton \cite{MR44:1546} and Walkoe \cite{MR43:4646} observed that exactly the properties definable in existential second-order logic (\eso) can be expressed
with partially ordered quantifiers.
The second step was taken by Hintikka and Sandu \cite{MR1034575,MR1410063}, who introduced independence-friendly logic, which extends \fo in terms of so-called slashed quantifiers.
For example, in 
\[\forall x_0\exists x_1\forall x_2\exists
x_3/\forall x_0\phi,\]
the quantifier $\exists x_3/\forall x_0$ means that
$x_3$ is ``independent'' of $x_0$ in the sense that a choice for
the value of $x_3$ should not depend on what the value of $x_0$
is. The semantics of \ifl was first formulated in game theoretic terms, and
\ifl can be regarded as a game theoretically motivated generalization of \fo.
Whereas the semantic game for \fo is a game of perfect information, the game for \ifl is a game of imperfect information.
The so-called team semantics of \ifl, also used in this paper, was introduced by Hodges  \cite{MR1465612}.     

Dependence logic, introduced by V\"a\"an\"anen \cite{va07},
was inspired by \ifl-logic, but the approach of V\"a\"an\"anen provided a fresh perspective on quantifier dependence. In dependence logic the 
dependence relations between variables are written in terms of novel 
atomic dependence formulas. For example, the partially ordered quantifier \eqref{poc} can be expressed in dependence logic as follows
\[ 
\forall x_0\exists x_1\forall x_2\exists
x_3(\dep(x_2,x_3)\wedge\phi).
\] 
The  atomic formula $\dep(x_2,x_3)$ has
the explicit meaning that $x_3$ is completely determined by $x_2$ and nothing else. 

In recent years, research related to \ifl and \df has been active. A variety of closely related logics have been defined and various applications suggested, see e.g.~\cite{Abramsky:2007, Bradfield:2005, Gradel+aananen:2010, lovo10, Sevenster:2009, Vaananen+Hodges:2010}.
While both \ifl and \df are known to be equi-expressive to \eso, the relative strengths and weaknesses of the two different logics in relation to applications is not understood well. In this article we take a step towards a better understanding of this matter.
After recalling some basic properties in Section \ref{preliminaries}, we compare the expressivity of the finite variable fragments of \df and \ifl in Section \ref{comparison}. We show that there is an effective translation from \dtwo to \iftwo (Theorem~\ref{dtoif}) and from \iftwo to $\df^3$ (Theorem~\ref{iftod}). 
We also show that \iftwo is strictly more expressive than \dtwo (Proposition~\ref{d less than if}). 
This result is a by-product of our proof in Section \ref{ifsatsection} that the satisfiability problem of \iftwo is undecidable (Theorem~\ref{iftwo sat complexity} 
shows \pizeroone-completeness). The proof can be adapted to the context of finite satisfiability, i.e., the problem of determining for a given formula $\phi$ whether there is a finite structure $\mA$ 
such that $\mA\models \phi$ (Theorem~\ref{iftwo finsat complexity} shows \sigmazeroone-completeness). The undecidability proofs are based on tiling arguments.
Finally, in Section \ref{sec:satd2}, we study the decidability of the satisfiability and finite satisfiability problems of \dtwo. 
For this purpose we reduce the problems to the (finite) satisfiability problem for \foctwo (Theorem~\ref{DtoESO}) and thereby show 
that they are \NEXPTIME-complete (Theorem~\ref{dtwo nexptime}). Table~\ref{table:results} gives an overview of previously-known as well as new complexity results.

\section{Preliminaries}\label{preliminaries}
In this section we recall the basic concepts and results relevant for this article.
%

The domain of a structure $\mA$ is denoted by $A$. We assume that the reader is familiar with first-order logic \fo. The extension of \fo in terms of counting quantifiers $\exists ^{\ge i}$ is denoted by \foc. We also consider the extension $\fo(\I)$ of \fo by the H\"artig quantifier $\I$. The interpretation of the quantifier \I is defined by the clause
\[\mA,s\models \I\, xy(\phi(x),\psi(
y))\Leftrightarrow |\phi(x)^{\mA,s}|=|\psi(y)^{\mA,s}|,  \]
where $\phi(x)^{\mA,s}:=\{a\in A\ |\ \mA,s\models \phi(a)\}$.
The $k$-variable fragments $\fo^k$, $\foc^k$, and $\fo^k(\I)$ are the fragments of \fo, \foc, and $\fo(\I)$ with formulas in which at most $k$, say $x_1,\ldots,x_k$, distinct variables appear. In the case $k=2$, we denote these variables by $x$ and $y$. 
The existential fragment of second-order logic is denoted by \eso.
 For logics $\calL$ and $\mathcal L'$, we write $\mathcal L \leq \mathcal L'$ if for  every sentence $\phi$ of $\mathcal L$ there is a sentence $\phi^*$ of $\mathcal L'$ such that for all structures $\mA$ it holds that $\mA\models \phi$ iff $\mA\models \phi^*$. We write $\calL \mathcal \equiv \mathcal L'$ if $\mathcal L \leq \mathcal L'$ and $\mathcal L' \leq \mathcal L$. 

We assume that the reader is familiar with the basics of computational complexity theory. In this article we are interested in the complexity of the satisfiability problems of various logics. For any logic \calL the \emph{satisfiability problem} $\sat[\calL]$ is defined as
\[\sat[\calL] := \{\phi \in \calL\mid \text{there is a structure $\mA$ such that $\mA\models \phi$}\}.\]

The finite satisfiability problem $\finsat[\calL]$ is the analogue of $\sat[\calL]$ in which we require the structure $\mA$ to be finite. The following observation will be useful later.

\begin{remark}\label{eso sat}
If $\phi$ is a formula over the vocabulary $\tau$ and
\[\psi := \exists R_1 \dots \exists R_n \exists f_1 \dots \exists f_m \phi\]
with $R_1,\dots,R_n,f_1,\dots,f_m\in \tau$, then $\phi$ is satisfiable iff the second-order formula $\psi$ is satisfiable.
\end{remark}

\subsection{\texorpdfstring{The logics \df and  \ifl}{The logics D and IF}}
In this section we define independence-friendly logic and dependence logic and recall some related basic results. For $\ifl$ we follow the exposition of \cite{CaicedoDJ09} and the forthcoming monograph \cite{ASS}.

\begin{definition}\label{def:iftwo intuitive}
The syntax of $\ifl$ extends the syntax of $\fo$ defined in terms
of $\vee$, $\wedge$, $\neg$, $\exists$ and $\forall$,
by adding quantifiers of the form
\begin{align*}
&\exists x/W\phi\\
&\forall x/W\phi
\end{align*}
called slashed quantifiers, where $x$ is a first-order variable, $W$ a finite set of first-order variables and $\phi$ a formula. 
\end{definition}

\begin{definition}[\cite{va07}]\label{def:dtwo intuitive}
The syntax of $\df$ extends the syntax of $\fo$, defined in terms
of $\vee$, $\wedge$, $\neg$, $\exists$ and $\forall$, by new
atomic (dependence) formulas of the form
\begin{equation*}
\dep[t_1,\ldots,t_n],
\end{equation*} where $t_1,\ldots,t_n$ are terms. 
\end{definition}

The set $\fr[\phi]$ of free variables of a formula 
$\phi\in \df \cup \ifl$ is defined  as for first-order logic except that we have the new cases
\[\begin{array}{lcl}
\fr(\dep(t_1,\ldots,t_n))=\var(t_1)\cup\cdots \cup \var(t_n)\\
\fr(\exists x /W\psi)=W\cup (\fr(\psi)\setminus \{x\})\\
\fr(\forall x /W\psi)=W\cup (\fr(\psi)\setminus \{x\})\\
\end{array}\]
where $\var(t_i)$ is the set of variables occurring in the term $t_i$. If $\fr(\phi)=\emptyset$, 
we call $\phi$ a sentence. 

\begin{definition}\label{def:diftwo}
Let $\tau$ be a relational vocabulary, i.e.,~$\tau$ does not contain function or constant symbols.
\begin{enumerate}[a)]
 \item 
The \emph{two-variable independence-friendly} logic $\iftwo(\tau)$ is generated from $\tau$ according to the following grammar:
\begin{align*}
\phi::= &t_1=t_2\mid R(t_1,\dots,t_n) \mid \neg t_1=t_2\mid \neg R(t_1,\dots,t_n)  \mid\\
&(\phi\wedge\phi) \mid (\phi\vee\phi) \mid \forall x \phi \mid \forall y \phi\mid \exists x/W\phi\mid \exists y/W\phi
\end{align*}
\item The \emph{two-variable dependence} logic $\dtwo(\tau)$ is generated from $\tau$ according to the following grammar:
\begin{align*}
\phi::= &t_1=t_2\mid R(t_1,\dots,t_n) \mid \neg t_1=t_2\mid \neg R(t_1,\dots,t_n)  \mid  \dep[t_1,t_2] \mid \neg \dep[t_1,t_2]\mid\\
& \dep[t_1] \mid \neg \dep[t_1]\mid(\phi\wedge\phi) \mid (\phi\vee\phi) \mid \forall x \phi \mid \forall y \phi\mid \exists x\phi\mid \exists y\phi
\end{align*}
\end{enumerate}
Here $R\in\tau$ is an $n$-ary relation symbol, $W\subseteq\{x,y\}$ and $t_1,\dots,t_n \in \{x,y\}$. We identify existential first-order quantifiers with existential quantifiers with empty slash sets, and therefore if $W=\emptyset$ we simply write $\exists x\phi(x)$ instead of $\exists x /W \phi(x)$.
When $\tau$ is clear we often leave it out. To simplify notation, we assume in the following that the relation symbols $R\in\tau$ are at most binary. 
\end{definition}

Note that in Definition~\ref{def:diftwo} we have only defined formulas in negation normal form and for that reason we do not need the slashed universal quantifier in \iftwo \cite{MR1465612}. Defining syntax in negation normal form is customary in $\ifl$ and \df. A formula $\phi$ with arbitrary negations is considered an abbreviation of the negation normal form formula $\psi$ obtained from $\phi$ by pushing the negations to the atomic level in the same fashion as in first-order logic.
It is important to note that the game theoretically motivated negation $\neg$ of \df and \ifl does not satisfy the law of excluded middle and is therefore not the classical Boolean negation. This is manifested by the existence of sentences  $\phi$ such that for some  $\mA$ we have $\mA\not\models\phi$ and $\mA\not\models\neg\phi$.

In order to define the semantics of $\ifl$ and $\df$, we first need to
define the concept of a \emph{team}. Let $\mA$ be a model with the domain $A$. {\em Assignments} over $\mA$
are finite functions that map variables to elements of $A$. The value of a term $t$
in an assignment $s$ is denoted by $t^{\mA}\langle s\rangle$.
If $s$ is an assignment, $x$ a variable,  and $a\in A$, then $s(a/x)$ denotes the 
assignment (with the domain $\dom(s)\cup \{x\}$)  which agrees with $s$
everywhere except that it maps $x$ to $a$.

Let $A$ be a set and $\{x_1,\ldots,x_k\}$ a finite (possibly empty) set  of
variables. 
 A {\em team} $X$ of $A$ with the domain
$\dom(X)=\{x_1,\ldots,x_k\}$ is any set of assignments from the variables
$\{x_1,\ldots,x_k\}$ into the set $A$. We denote by $\rel(X)$ the
$k$-ary relation of $A$ corresponding to $X$ 
\[\rel(X)=\{(s(x_1),\ldots,s(x_k)) \mid s\in X \}.  \]
 If $X$ is a team of $A$,
and $F\colon X\rightarrow A$, we use $X(F/x)$ to denote the team
$\{s(F(s)/x) \mid s\in X \}$ and $X(A/x)$ the team $\{s (a/x)
\mid s\in X\ \textrm{and}\ a\in A \}$.
For a set $W\subseteq \dom(X)$ we call $F$ \emph{$W$-independent} if for all $s,s' \in X$ with $s(x)=s'(x)$ for all $x\in \dom(X)\setminus W$ we have that $F(s)=F(s')$.

We are now ready to define the semantics of \ifl and \df.

\begin{definition}[\cite{MR1465612,va07}]\label{def:semantics}
Let $\mA$ be a model and $X$ a team of $A$. The satisfaction relation
$\mA\models _X \phi$ is defined as follows:
\begin{enumerate}
\item If $\phi$ is a first-order literal, then $\mA\models_X \phi$ iff for all $s\in X$: $\mA,s\models_{\fo}\phi$.

\item $\mA\models_X \psi \wedge \phi$ iff $\mA\models _X \psi$ and $\mA\models _X \phi$.
\item $\mA\models_X \psi \vee \phi$ iff there exist teams $Y$ and $Z$ such that $X=Y\cup Z$, $\mA\models_Y \psi$ and $\mA\models _Z \phi$.

\item   $\mA \models_X \exists x \psi$ iff $\mA \models _{X(F/x)} \psi$ for some $F\colon X\to A$.

\item $\mA \models_X \forall x\psi$ iff $\mA \models _{X(A/x)} \psi$.
\end{enumerate}
For \ifl we further have the following rules:
\begin{enumerate}
\setcounter{enumi}{5}
\item $\mA\models_X \exists x/W\phi$ iff $\mA\models_{X(F/x)}\phi$ for some $W$-independent function $F:X\rightarrow A$.

\item $\mA\models_X \forall x/W\phi$ iff $\mA\models_{X(A/x)}\phi$.

\end{enumerate}
And for \df we have the additional rules:
\begin{enumerate}
\setcounter{enumi}{7}
\item $\mA\models_X \dep(t_{1},\ldots,t_{n})$ iff for all $s,s'\in
X$ such that $t_1^{\mA}\langle s\rangle  =t_1^{\mA}\langle
s'\rangle  ,\ldots, t_{n-1}^{\mA}\langle s\rangle
=t_{n-1}^{\mA}\langle s'\rangle  $, we have $t_n^{\mA}\langle
s\rangle  =t_n^{\mA}\langle s'\rangle  $.

\item  $\mA \models_X \neg \dep(t_{1},\ldots,t_{n})$ iff
$X=\emptyset$.
\end{enumerate}
Above, we assume that the domain of $X$ contains $\fr[\phi]$. Finally, a sentence $\phi$ is true in a model $\mA$  ($\mA\models \phi$)  if $\mA\models _{\{\emptyset\}} \phi$.
\end{definition}

From Definition \ref{def:semantics} it follows that many familiar propositional equivalences of connectives do not hold in \df and \ifl. For example, the idempotence of disjunction fails, which can be used to show that the  distributivity 
laws of disjunction and conjunction do not hold either. We refer
to \cite[Section~3.3]{va07} for a detailed exposition on propositional equivalences of connectives in \df (and also \ifl). Another feature of 
Definition \ref{def:semantics} is that $\mA\models _{\emptyset}\phi$ for all $\mA$ and all formulas $\phi$ of \df and \ifl. This observation is important in noting that, for sentences $\phi$ and $\psi$, the interpretation  of $\phi\vee \psi$ coincides with the classical disjunction of $\phi$ and $\psi$.


\subsection{\texorpdfstring{Basic properties of \df and \ifl}{Basic properties of D and IF}}
In this section we recall some basic properties of \df and \ifl.

Let $X$ be a team with the domain $\{x_1,\ldots,x_k\}$ and $V\subseteq \{x_1,\ldots,x_k\}$. We denote by
$X\upharpoonright V$ the team $\{s\upharpoonright V \mid s\in X\}$ with the domain $V$. The following proposition
 shows that
the truth of a \df-formula  depends only on the interpretations of the variables occurring free in the formula.
\begin{proposition}[\cite{va07, CaicedoDJ09}]\label{freevar} Let $\phi\in \df$ be any formula or $\phi\in \ifl$ a sentence. If 
 $V\supseteq \fr(\phi)$, then $\mA \models _X\phi$ if and only if $\mA \models _{X\upharpoonright V} \phi$.
\end{proposition}
The analogue of Proposition~\ref{freevar} does not hold for open formulas of \ifl. In other words, the truth of an $\ifl$-formula may depend on the interpretations of variables that do not occur
in the formula. For example, the truth of the formula $\phi$
\begin{equation*}
\phi = \exists x/\{y\}(x=y)
\end{equation*}
 in a team $X$ with domain $\{x,y,z\}$ depends on the values of $z$ in $X$, although $z$ does not occur in $\phi$. 

The following fact is a fundamental property of all formulas of \df and \ifl:

\begin{proposition}[\cite{va07, MR1465612}, Downward closure]\label{Downward closure}
Let $\phi$ be a formula of  \df or \ifl, $\mA$  a model, and $Y\subseteq X$ teams. Then $\mA\models_X \phi$ implies $\mA\models_Y\phi$. 
\end{proposition}

The expressive
power of sentences of $\df$ and \ifl coincides with that of existential
second-order sentences:

\begin{theorem}\label{d equiv if equiv eso}
$\df\equiv \ifl\equiv \eso$.
\end{theorem}
\begin{proof}The fact $\eso\le \df$ (and $\eso\le \ifl$) is based on the analogous result of \cite{MR44:1546,MR43:4646} for partially ordered quantifiers.  For the converse inclusions, see  \cite{va07} and  \cite{MR1638352}.
\end{proof}

\begin{proposition}[\cite{va07, MR1465612}]\label{FO extension}
Let $\phi$ be a formula of \df or \ifl without dependence atoms and without slashed quantifiers, i.e.,~$\phi$ is syntactically a first-order formula. Then for all $\mA$, $X$ and $s$:
\begin{enumerate}
\item $\mA\models _{\{s\}}\phi$ iff $\mA,s\models_{\fo}\phi$.
\item $\mA\models _X\phi$ iff for all $s\in X:\, \mA,s\models_{\fo}\phi$.
\end{enumerate}
\end{proposition}

\section{\texorpdfstring{Comparison of \ifl and \df}{Comparison of IF and D}}\label{comparison}

In this section we show that
\[\dtwo< \iftwo\leq \df^3.\]
We also further discuss the expressive powers and other logical properties of $\dtwo$ and $\iftwo$.

\begin{lemma}\label{dtoif} For any formula $\phi\in \dtwo$
 there is a formula $\phi^*\in \iftwo$ such that for all structures
  $\mA$ and  teams $X$, where $\dom(X)=\{x,y\}$, it holds that  
\begin{eqnarray*}
 \mA\models _X\phi &\Leftrightarrow& \mA\models _X\phi^*.
 \end{eqnarray*}
\end{lemma}
\begin{proof}
The translation $\phi\mapsto \phi^*$ is defined as follows. For  first-order literals the translation is the identity, and negations of dependence atoms are translated by $\neg x=x$. The remaining cases are defined as follows:
\begin{eqnarray*}
\dep(x) &\mapsto& \exists y/\{x,y\}(x=y)\\
\dep(x,y) &\mapsto& \exists x/\{y\}(x=y)\\
\phi \wedge \psi   &\mapsto& \phi^* \wedge \psi^*\\
\phi \vee \psi   &\mapsto& \phi^* \vee \psi^*\\
\exists x \phi   &\mapsto& \exists x  \phi^* \\
\forall  x \phi   &\mapsto& \forall x  \phi^*
\end{eqnarray*}
 The claim of the lemma can now be proved using induction on $\phi$.
The only non-trivial cases are the dependence atoms. We consider the case
where $\phi$ is of the form $\dep(x,y)$. 

Let us assume that $\mA\models_X \phi$. Then there is a function $F\colon A\to A$ such that
\begin{equation}\label{dtoif6}
\text{for all }s\in X:\ s(y)=F(s(x)).
\end{equation}
Define now  $F'\colon X \to A$ as follows: 
\begin{equation}\label{dtoif3}
F'(s):=F(s(x)).
\end{equation} 
$F'$ is  $\{y\}$-independent since, if $s(x)=s'(x)$, then
\[F'(s)= F(s(x))
{=}F(s'(x))=F'(s').\]
It remains to show that 
\begin{equation}\label{dtoif2}
 \mA \models_{X(F'/x)}(x=y).
\end{equation}
Let $s\in X(F'/x)$. Then
\begin{equation}\label{dtoif5}
s= s'(F'(s')/x)\text{ for some }s'\in X.
\end{equation}
Now
\[s(x)\stackrel{\eqref{dtoif5}}{=}F'(s')\stackrel{\eqref{dtoif3}}{=}F(s'(x))\stackrel{\eqref{dtoif6}}{=}s'(y)\stackrel{\eqref{dtoif5}}{=}s(y).\]
Therefore, \eqref{dtoif2} holds, and hence also 
\[ \mA \models_X \exists x/\{y\}(x=y).\]

Suppose then that $\mA\not \models_X \phi$. Then there must be $s,s'\in X$ such that $s(x)=s'(x)$ and $s(y)\neq s'(y)$.
We claim now that 
\begin{equation}\label{dtoif1} 
\mA\not \models_X \exists x/\{y\}(x=y).
\end{equation}
Let $F\colon X\rightarrow A$ be an arbitrary $\{y\}$-independent function. Then, by 
$\{y\}$-independence, $F(s)=F(s')$ and since additionally $s(y)\neq s'(y)$, we have
\[s(F(s)/x)(x)=F(s)\neq s(y)=s(F(s)/x)(y)\]
or
\[s'(F(s')/x)(x)=F(s')\neq s'(y)=s'(F(s')/x)(y).\]
This implies that
\[\mA \not\models_{X(F/x)}(x=y),\]
since $s(F(s)/x),\,s'(F(s')/x) \in X(F/x)$.

Since $F$ was arbitrary, we may conclude that \eqref{dtoif1} holds.
\end{proof}

Next we show a translation from $\iftwo$ to $\df^3$.
\begin{lemma}\label{iftod} For any formula $\phi\in \iftwo$
 there is a formula $\phi^*\in \df^3$ such that for all
 structures  $\mA$ and  teams $X$, where $\dom(X)=\{x,y\}$, it holds that  
\begin{eqnarray*}
 \mA\models _X\phi &\Leftrightarrow& \mA\models _X\phi^*.
 \end{eqnarray*}
\end{lemma}
\begin{proof}
The claim follows by the following translation $\phi\mapsto \phi^*$:
For atomic and negated atomic formulas the translation is the identity, and for propositional connectives and first-order quantifiers it is defined in the obvious inductive way. The only non-trivial cases are the slashed quantifiers:
\[\begin{array}{rcl}
\exists x /\{y\}\psi  &\mapsto&  \exists z(x=z\wedge \exists x (\dep(z,x) \wedge \psi^*)),\\
\exists x /\{x\}\psi  &\mapsto&  \exists x (\dep(y,x) \wedge \psi^*),\\
\exists x /\{x,y\}\psi  &\mapsto&  \exists x(\dep(x)\wedge \psi^*).\\
\end{array}\]
Again, the claim can be proved using induction on $\phi$.
We consider the case where $\phi$ is of the form $\exists x /\{y\}\psi$.
Assume $\mA\models _X \phi$. Then there is a $\{y\}$-independent function $F\colon X\rightarrow A$ such that
\begin{equation}\label{iftod0}
\mA\models_{X(F/x)} \psi. 
\end{equation}
By $\{y\}$-independence, $s(x)=s'(x)$ implies that $F(s)=F(s')$ for all $s,s'\in X$. Our goal is to show that 
\begin{equation}\label{iftod1}
\mA\models_X \exists z(x=z\wedge \exists x (\dep(z,x) \wedge \psi^*)).
\end{equation}
Now, \eqref{iftod1} holds if for $G\colon X \to A$ defined by $G(s)=s(x)$ for all $s\in X$ it holds that
\begin{equation}\label{iftod2}
\mA\models_{X(G/z)} \exists x (\dep(z,x) \wedge \psi^*).
\end{equation}
Define $F'\colon X(G/z)\rightarrow A$ by
$F'(s)=F(s\upharpoonright\{x,y\})$. Now we claim that
\[ \mA\models_{X(G/z)(F'/x)} \dep(z,x) \wedge \psi^*,\]
implying \eqref{iftod2} and hence \eqref{iftod1}.

First we show that 
\begin{equation}\label{iftod42}
\mA\models_{X(G/z)(F'/x)} \dep(z,x).
\end{equation}
At this point it is helpful to note that every $s\in {X(G/z)(F'/x)}$ arises from an $s'\in X$ by first copying the value of $x$ to $z$ and then replacing the value of $x$ by $F(s\upharpoonright\{x,y\})$, i.e.,~that $s(z)=s'(G(s')/z)(z)=G(s')=s'(x)$ and $s(x)=F(s')$.
Now, to show \eqref{iftod42}, let $s_1,s_2\in X(G/z)(F'/x)$ with $s_1(z)=s_2(z)$ and let $s'_1,s'_2\in X$ as above, i.e., $s_1$ (resp.~$s_2$) arises from $s'_1$ (resp.~$s'_2$). Then it follows that $s'_1(x)=s'_2(x)$. Hence, by $\{y\}$-independence, $F(s'_1)=F(s'_2)$, implying that $s_1(x)=F(s'_1)=F(s'_2)=s_2(x)$ which proves \eqref{iftod42}.
Let us then show that  
\begin{equation}\label{iftod3}
 \mA\models_{X(G/z)(F'/x)}\psi^*.
\end{equation}
Note first that by the definition of the mapping $\phi\mapsto \phi^*$ the variable $z$ cannot appear free in $\psi^*$. By Proposition \ref{freevar}, the satisfaction of any $\df$-formula $\theta$ only depends on those variables in a team that appear free in $\theta$, therefore \eqref{iftod3} holds iff
 \begin{equation}\label{iftod4}
 \mA\models_{X(G/z)(F'/x)\upharpoonright\{x,y\}}\psi^*.
\end{equation}
We have chosen  $G$ and $F'$ in such a way that 
$$X(G/z)(F'/x)\upharpoonright\{x,y\}=X(F/x),$$
hence \eqref{iftod4} now follows from \eqref{iftod0} and the induction hypothesis.

We omit the proof of the converse implication which is analogous.
\end{proof}

For sentences, Lemmas \ref{dtoif} and \ref{iftod} 
now imply the following.

\begin{theorem}\label{dtwo le if2}
$\dtwo\le \iftwo\le \df^3$
\end{theorem}
\begin{proof} 
The claim follows by Lemmas \ref{dtoif} and \ref{iftod}.
 First of all, if $\phi$ is a sentence of $\ifl$ or $\df$, 
then, by Proposition \ref{freevar}, for every model $\mA$ and team $X\neq \emptyset$ 
\begin{equation}\label{sentences}
\mA \models_X \phi\text{ iff } \mA \models_{\{\emptyset\}} \phi.
\end{equation}
It is important to note that, even if $\phi\in \dtwo$ is a sentence, it
 may happen that $\phi^*$ has free variables since variables
 in $W$ are regarded as free in subformulas of $\phi^*$ of the form
 $\exists x /W\psi$. However, this is not a problem. Let $Y$ be the set of all assigments of $\mA$ with the domain $\{x,y\}$. Now
\begin{eqnarray*}
 \mA\models_{\{\emptyset\}}\phi &\textrm{ iff }&   \mA\models_Y\phi \textrm{ iff } \mA\models_Y \forall x\forall y \phi \\ &\textrm{ iff }&  \mA\models_Y\forall x\forall y \phi^* \text{ iff } \mA\models_{\{\emptyset\}}\forall x\forall y \phi^* ,
\end{eqnarray*}
where the first and the last equivalence hold by \eqref{sentences}, the second by the semantics of the universal quantifier and the third by Lemma \ref{dtoif}. An analogous argument can be used to show that for every sentence $\phi\in \iftwo$ there is an equivalent sentence of the logic $\df^3$.
%
%
\end{proof}

\subsection{\texorpdfstring{Examples of properties definable in \dtwo}{Examples of properties definable in D\textasciicircum2}}
We end this section with examples of definable classes of structures in \dtwo (and in \iftwo by Theorem \ref{dtwo le if2}).

\begin{proposition}

The following properties can be expressed in \dtwo:
\begin{enumerate}[a)]
\item\label{example1} For unary relation symbols $P$ and $Q$, \dtwo can express $|P|=|Q|$. This shows $\dtwo \not \le \fo$.
\item\label{example2} If the vocabulary of $\mA$ contains a constant $c$, then \dtwo can express that $A$ is infinite.
\item\label{example3} $|A|\le k$ can be expressed already in $\df^1$.
\end{enumerate}
\end{proposition}
\begin{proof}
Let us first consider part \ref{example1}). Clearly, it suffices to express $|P|\le |Q|$. Define $\phi$ by
\[\phi:= \forall x\exists y( \dep(y,x)\wedge (\neg P(x)\vee Q(y))).\]
Now,  $\mA\models \phi$\quad iff\quad there is an injective function $F\colon A\rightarrow A$ such that $F[P^{\mA}]\subseteq Q^{\mA}$ \quad iff\quad $|P^{\mA}|\le |Q^{\mA}|$.  

For part \ref{example2}), we use the same idea as above. Define $\psi$ by
\[\psi:= \forall x\exists y(\dep(y,x)\wedge \neg c=y).\]
Now, $\mA\models\psi$\quad iff\quad  there is an injective function $F\colon A\rightarrow A$ such that $c^{\mA}\notin F[A]$\quad iff\quad $A$ is infinite.

Finally, we show how to express the property from part \ref{example3}). Define $\theta$ as
\[\forall x( \bigvee_{1\le i\le k}\chi_i),    \]
where $\chi_i$ is $\dep(x)$. It is now immediate that $\mA\models \theta$ iff $|A|\le k$.
\end{proof}

It is interesting to note that, although part \ref{example1}) holds, the difference in \sat-complexity of $\fotwo(\I)$ and $\dtwo$ is a major one. The former is \sigmaoneone-hard \cite{grotro97} whereas the latter is decidable -- as is shown in section \ref{sec:satd2}. Part \ref{example1}) also implies that \dtwo does not have a zero-one law, since the property $|P|\le |Q|$ (which can be expressed in \dtwo) has the limit probability $\frac{1}{2}$.

\begin{proposition}\label{d less than if} 
$\dtwo <\iftwo$. This holds already in the finite.
\end{proposition}
\begin{proof} The property of being  grid-like  (see Definition \ref{def:gridlike})
can be expressed in \iftwo but not in \dtwo since \dtwo is decidable by  Theorem \ref{dtwo nexptime}. 
In the finite, there exists no \dtwo sentence equivalent to the
\iftwo sentence $\phi_{\mathrm{torus}}$ (see Section \ref{iffinsatsection}),
since the finite satisfiability problem of \dtwo is decidable.
\end{proof}




\section{\texorpdfstring{Satisfiability for \iftwo is undecidable}{Satisfiability for IF\textasciicircum2 is undecidable}}\label{ifsatsection}
In this section we will use tiling problems, introduced by Hao Wang in \cite{Wang:1961}, to show the undecidability of \sat[\iftwo] as well as \finsat[\iftwo].
%

In this paper a \emph{Wang tile} is a square in which each edge is assigned a color. It is a square that has four colors (up, right, down, left).
We say that a set of tiles can tile the $\N\times\N$ plane if a tile can be placed on every point $(i,j)\in \N\times\N$ s.t.~the right color of the tile in $(i,j)$ is the same as the left color of the tile in $(i+1,j)$ and the up color of the tile in $(i,j)$ is the same as the down color in the tile in $(i,j+1)$.
Notice that turning and flipping tiles is not allowed.

We then define some specific structures needed later.

\begin{definition}\label{def:grid}
The model $\mG:=(G,V,H)$ where
\begin{itemize}
 \item $G=\N\times\N$,
 \item $V=\{((i,j),(i,j+1))\subseteq G\times G\mid i,j\in\N\}$ and
 \item $H=\{((i,j),(i+1,j))\subseteq G\times G\mid i,j\in\N\}$
\end{itemize}
is called the \emph{grid}.

A finite model $\mD=(D,V,H,V',H')$ where
\begin{itemize}
 \item $D=\{0,\dots, n\}\times\{0,\dots,m\}$,
 \item $V=\{((i,j),(i,j+1))\subseteq D\times D\mid i\leq n,j<m\}\}$,
 \item $H=\{((i,j),(i+1,j))\subseteq D\times D\mid i<n,j\leq m\}$,
 \item $V'=\{((i,m),(i,0))\subseteq D\times D\mid i\leq n\}$ and
 \item $H'=\{((n,j),(0,j))\subseteq D\times D\mid j\leq m\}$
\end{itemize}
is called a \emph{torus}.
\end{definition}

\begin{definition}\label{def:tiling}
A set of \emph{colors} $C$ is defined to be an arbitrary finite subset of the natural numbers. The set of all \emph{(Wang) tiles} over $C$ is $C^4$, i.e.,~a tile is an ordered list of four colors, interpreted as the colors of the four edges of the tile in the order top, right, bottom and left.

Let $C$ be a set of colors, $T\subseteq C^4$ a finite set of tiles and $\mA=(A,V,H)$ a first-order structure with binary relations $V$ and $H$ interpreted as vertical and horizontal successor relations. Then a $T$-\emph{tiling} of $\mA$ is a total function $t\colon A \to T$ such that for all $x,y\in A$ it holds that
\begin{enumerate}[i)]
 \item $(t(x))_0=(t(y))_2$ if $(x,y)\in V$, i.e.,~the top color of $x$ matches the bottom color of $y$, and
 \item $(t(x))_1=(t(y))_3$ if $(x,y)\in H$, i.e.,~the right color of $x$ matches the left color of $y$.
\end{enumerate}
\end{definition}

Next we define the tiling problem for a structure $\mA=(A,V,H)$. 
\begin{definition}[\tiling]
A structure $\mA=(A,V,H)$ is called $T$-\emph{tilable} iff there is a $T$-tiling of $\mA$.

For any structure $\mA=(A,V,H)$ we define the problem
\[\tiling[\mA] := \{T\mid \text{$\mA$ is $T$-tilable}\}.\]

We say that a structure $\mB=(B,V,H,V',H')$ is $T$-tilable if and only if the structure $(B,V\cup V',H\cup H')$ is $T$-tilable.
Hence a torus $\mD=(D,V,H,V',H')$ is $T$-tilable if and only if the structure $(D,V\cup V',H\cup H')$ is $T$-tilable.
Now we define the problem
\[\tiling[\torus] := \{T\mid \text{there is a torus $\mD$ that is $T$-tilable}\}.
\]
\end{definition}
Note that the set \tiling[\mG] consists of all $T$ such that there is a $T$-tiling of the infinite grid and \tiling[\torus] consists of all $T$ such that there is a \emph{periodic} $T$-tiling of the grid.
Further note that \tiling[\torus] cannot be expressed in the form \tiling[\mD] for a fixed torus \mD since a fixed torus has a fixed size and we want the problem to be the question whether there is a torus of \emph{any} size.

We will later use the following two theorems to show the undecidability of \sat[\iftwo] and, resp., \finsat[\iftwo].

\begin{theorem}[\cite{be66}, \cite{Harel:1986}]\label{tiling pizeroone}
\tiling[\mG] is $\pizeroone$-complete.
\end{theorem}

\begin{theorem}[{\cite[Lemma~2]{guko72}}]\label{periodic tiling pizeroone}
\tiling[\torus] is $\sigmazeroone$-complete.
\end{theorem}

To prove the undecidability of \sat[\iftwo] (Theorem~\ref{iftwo sat complexity}) we will, for every set of tiles $T$, define a formula $\phi_T$ such that $\mA\models \phi_T$ iff $\mA$ has a $T$-tiling. Then we will define another formula $\phi_\mathrm{grid}$ and show that $\mA\models \phi_\mathrm{grid}$ iff $\mA$ contains (an isomorphic copy of) the grid as a substructure.
Therefore $\phi_T\wedge \phi_\mathrm{grid}$ is satisfiable if and only if there is a $T$-tiling of the grid.
For the undecidability of \finsat[\iftwo] (Theorem~\ref{iftwo finsat complexity}) we will define a formula $\phi_\mathrm{torus}$ which is a modification of the formula $\phi_\mathrm{grid}$.

\begin{definition}\label{def:tiling formula}
Let $T=\{t^0,\dots,t^k\}$ be a set of tiles, and
for all $i\leq k$, let $\cmd{right}(t^i)$ (resp.~$\cmd{top}(t^i)$) be the set
\[\{t^j\in\{0,\dots,k\}\mid t^i_1 = t^j_3\text{ (resp.~}t^i_0 = t^j_2)\},\]
i.e.,~the set of tiles matching $t^i$ to the right (resp.~top).

Then we define the first-order formulas
\[\begin{array}{r@{}l}
\psi_T:=\forall x \forall y \bigg(\Big(&H(x,y)\rightarrow \bigwedge\limits_{i\leq k}\big(P_i(x)\rightarrow \bigvee\limits_{t^j\in \cmd{right}(t^i)} P_j(y)\big)\Big)\ \wedge\\
\Big(&V(x,y)\rightarrow \bigwedge\limits_{i\leq k}\big(P_i(x)\rightarrow \bigvee\limits_{t^j\in \cmd{top}(t^i)} P_j(y)\big)\Big)\bigg),\\[3ex]
\multicolumn{2}{l}{\theta_T:=\forall x \bigvee\limits_{i\leq k}\big(P_i(x)\wedge \bigwedge\limits_{\substack{j\leq k\\j\neq i}} \neg P_j(x)\big)\text{ and}}\\
\multicolumn{2}{l}{\phi_T:=\psi_T \wedge \theta_T,}
\end{array}\]
over the vocabulary $V,H,P_0,\dots,P_k$. In an \ifl or \df context, $\phi \rightarrow \psi$ is considered to be an abbreviation of $\phi^\neg \vee \psi$, where $\phi^\neg$ is the negation normal form of $\neg \phi$.
\end{definition}

\begin{lemma}\label{tiling iff tiling formula}
Let $T=\{t_0,\dots,t_k\}$ be a set of tiles and $\mA=(A,V,H)$ a structure. Then $\mA$ is $T$-tilable iff there is an expansion $\mA^*=(A,V,H,P_0,\dots,P_k)$ of $\mA$ such that $\mA^*\models\phi_T$.
\end{lemma}

\begin{lemma}\label{torustiling iff tiling formula}
Let $T=\{t_0,\dots,t_k\}$ be a set of tiles and $\mB=(B,V,H,V',H')$ a structure. There is an \fotwo sentence $\gamma_T$ of the vocabulary $\{V,H,V',H',P_0,\dots,P_k\}$ such that $\mB$ is $T$-tilable iff there is an expansion $\mB^*=(A,V,H,V',H',P_0,\dots,P_k)$ of $\mB$ such that $\mB^*\models\gamma_T$.
\end{lemma}

Notice that $\phi_T$ is an $\fotwo$-sentence. Therefore $T$-tiling is expressible even in $\fotwo$.
The difficulty lies in expressing that a structure is (or at least contains) a grid. This is the part of the construction where \fotwo or even \dtwo formulas are no longer sufficient and the full expressivity of \iftwo is needed.

\begin{definition}\label{def:gridlike}
A structure $\mA=(A,V,H)$ is called \emph{grid-like} iff
it satisfies the conjunction $\phi_\mathrm{grid}$ of the formulas
\[\begin{array}{lcl} 
\phi_{\mathrm{functional}}(R)            & :=    & \forall x\forall y \big(R(y,x)\,\to\, \exists y/\{x\}\,x=y \big)\\
                                &       & \quad\text{for $R\in\{V,H\}$},\\
\phi_{\mathrm{injective}}(R)             & :=    & \forall x\forall y \big(R(x,y)\, \to\, \exists y/\{x\} \,x=y \big)\\
                                &       & \quad\text{for $R\in\{V,H\}$},\\
\phi_{\mathrm{root}}                     & :=    & \exists x\forall y \big(\neg V(y,x)\wedge \neg H(y,x)\big),\\
\phi_{\mathrm{distinct}}                 & :=    & \forall x\forall y\, \neg \big(V(x,y)\wedge H(x,y)\big),\\
\phi_{\mathrm{edge}}(R,R')               & :=    & \forall x \Big(\big(\forall y \,\neg R(y,x)\big)\,\to \forall y \big(R'(x,y)\to \forall x\,\neg R(x,y)\big)\Big)\\
                                &       & \quad\text{for $(R,R')\in\{(V,H),(H,V)\}$,}\\
\phi_{\mathrm{join}}                     & :=    & \forall x \forall y \Big(\big(V(x,y)\vee H(x,y)\big)\,\to \exists x/\{y\}\, \big(V(y,x)\vee H(y,x)\big)\Big),\\
\phi_{\mathrm{infinite}}(R)              & :=    & \forall x\exists y R(x,y) \text{ for $R\in\{V,H\}$}.
\end{array}\]
\end{definition}

The grid-likeness of a structure can alternatively be described in the following more intuitive way.

\begin{remark}\label{gridlike description}
A structure $\mA=(A,V,H)$ is grid-like iff
\begin{enumerate}[i)]
\item\label{degree} $V$ and $H$ are (graphs of) injective total functions, i.e.,~the out-degree of every element is exactly one and the in-degree at most one ($\phi_{\mathrm{infinite}}$, $\phi_{\mathrm{functional}}$ and $\phi_{\mathrm{injective}}$),
\item there is an element, called the \emph{root}, that does not have any predecessors ($\phi_{\mathrm{root}}$),
\item for every element, its $V$ successor is distinct from its $H$ successor ($\phi_{\mathrm{distinct}}$),
\item for every element $x$ such that $x$ does not have a $V$ (resp.~$H$) predecessor, the $H$ (resp.~$V$) successor of $x$ also does not have a $V$ (resp.~$H$) predecessor ($\phi_{\mathrm{edge}}$),
\item\label{join description} for every element $x$ there is an element $y$ such that $(x,y) \in (V\circ H)\cap (H\circ V)$ or $(x,y) \in (V\circ V)\cap (H\circ H)$,
\end{enumerate}
\end{remark}
\begin{proof}
We show that a structure $\mB \models \phi_{\mathrm{grid}}$ satisfies the above five properties.
The only difficult case is property \ref{join description}). First note that $\phi_\mathrm{join}$ is equivalent to the first-order formula
\[\forall x \exists x' \forall y \Big(\big(V(x,y)\vee H(x,y)\big)\,\to \, \big(V(y,x')\vee H(y,x')\big)\Big).\]
Since $\phi_{\mathrm{functional}}$, $\phi_{\mathrm{distinct}}$ and $\phi_{\mathrm{infinite}}$ hold as well, $\mB$ satisfies 
\[
\forall x \exists x' \exists y_1 \exists y_2 \Big(y_1\neq y_2 \wedge V(x,y_1)\wedge H(x,y_2) \wedge \big(V(y_1,x')\vee H(y_1,x')\big) \wedge \big(V(y_2,x')\vee H(y_2,x')\big)\Big).
\]
Due to $\phi_{\mathrm{injective}}$, neither $V(y_1,x') \wedge V(y_2,x')$ nor $H(y_1,x') \wedge H(y_2,x')$ can be true if $y_1\neq y_2$. Hence, $\mB$ satisfies
\[\begin{array}{l}
\forall x \exists x' \exists y_1 \exists y_2 \Big(y_1\neq y_2 \wedge \Big(\big(V(x,y_1)\wedge H(x,y_2) \wedge V(y_1,x') \wedge H(y_2,x')\big) \vee\\
\quad \big(V(x,y_1)\wedge H(x,y_2) \wedge H(y_1,x') \wedge V(y_2,x')\big)\Big)\Big).
\end{array}\]
From this formula the property \ref{join description}) is immediate (with $x:=x$ and $y:=x'$).
\end{proof}

Now we will use Remark~\ref{gridlike description} to show that a grid-like structure, although it need not be the grid itself, must at least contain an isomorphic copy of the grid as a substructure.

\begin{theorem}\label{gridlike includes grid}
Let $\mA=(A,V,H)$ be a grid-like structure. Then $\mA$ contains an isomorphic copy of $\mG$ as a substructure.
\end{theorem}
\begin{proof}
If $\mB$ is a model with two binary relations $R$ and $R'$, $b\in B$ and $i\in \N$ then the \emph{$i$-$b$-generated substructure} of $\mB$ (denoted by $\mB^i(b)$) is defined inductively in the following way:
\[\begin{array}{lcl}
\mB^0(b)        &=& \mB\upharpoonright\{b\},\\
\mB^{i+1}(b)    &=& \mB\upharpoonright\big(B^i(b)\,\cup \{x\in B\mid \exists y\in B^i(b):(y,x)\in R\cup R'\}\big).
\end{array}\]

Let $r\in A$ be a root of $\mA$ (which exists because $\mA\models \phi_{\mathrm{root}}$). We call a point $a\in\mA$ a \emph{west border point} (resp.~\emph{south border point}) if $(r,a)\in V^n$ (resp.~$(r,a)\in H^n$) for some $n\in\N$. 
Due to Remark~\ref{gridlike description}, every point in $\mA$ has $V$- and $H$-in-degree at most one while the west border points have $H$-in-degree zero and the south border points have $V$-in-degree zero. We call a substructure $\mH$ of $\mA$ \emph{in-degree complete} if every point in $\mH$ has the same in-degrees in $\mH$ as it has in $\mA$.

We will prove by induction that there exists a family of isomorphisms $\{f_i\mid i\in\N\}$ such that
\begin{enumerate}
\item $f_i$ is an isomorphism from $\mG^i((0,0))$ to $\mA^i(r)$,
\item $\mA^i(r)$ is in-degree complete and
\item $f_{i-1}\subseteq f_{i}$
\end{enumerate}
for all $i\in\N$.

The basis of the induction is trivial. Clearly the function $f_0$ defined by $f_0((0,0)):=r$ is an isomorphism from $\mG^0((0,0))$ to $\mA^0(r)$. And since $r$ is a root it has no $V$- or $H$-predecessors. Hence, $\mA^0(r)$ is in-degree complete. 

\begin{figure}
\centering
\includegraphics{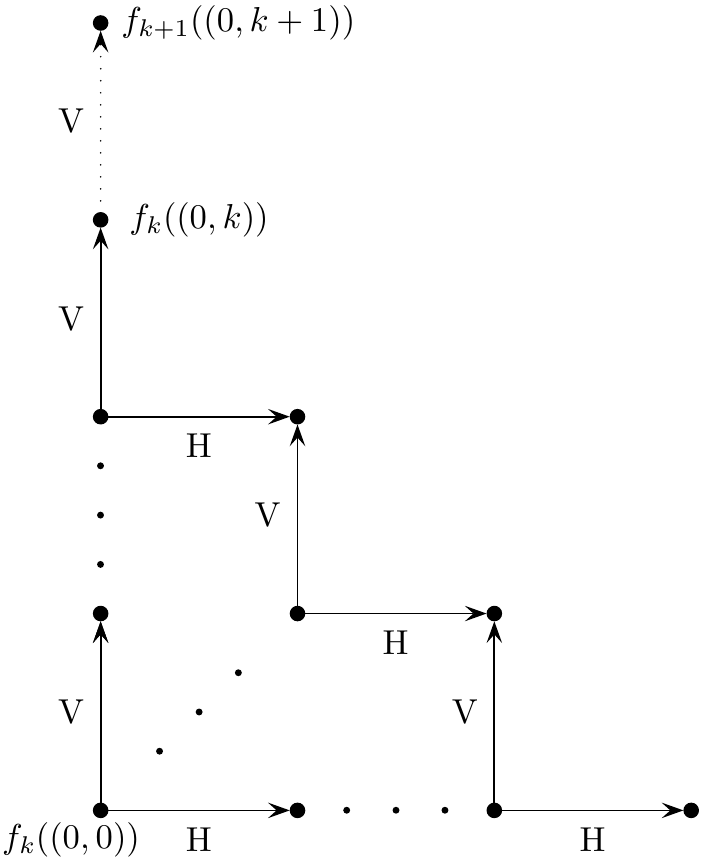}
\caption{The inductively defined substructures}
\label{fig:submodel}
\end{figure}

Let us then assume that $f_k$ is an isomorphism from $\mG^k((0,0))$ to $\mA^k(r)$, $\mA^k(r)$ is in-degree complete and $f_{k-1}\subseteq f_k$.
Then the $k$-$r$-generated substructure of $\mA^{k+1}(r)$ (which is $\mA^k(r)$) is isomorphic to $\mG^k((0,0))$ and the isomorphism is given by $f_k$.

We will now show how to extend $f_k$ to the isomorphism $f_{k+1}$. This is done by extending $f_k$ element by element along the diagonal (Figure~\ref{fig:submodel} shows the first extension step). We will abuse notation and denote the extensions of the function $f_k$ by $h$ throughout the proof. We will show by induction on $j$ that we can extend the isomorphism by assigning values for $h(j,(k+1)-j)$ for all $0\leq j\leq k+1$ -- still maintaining the isomorphism between $\mG\upharpoonright\dom(h)$ and $\mA\upharpoonright \rng[h]$, and the in-degree completeness of $\mA\upharpoonright \rng[h]$.

Due to $\phi_{\mathrm{infinite}}$ and $\phi_{\mathrm{functional}}$ the west border point $f_{k}((0,k))$ has a unique $V$-successor $a$. Since the $k$-$r$-generated substructure of $\mA^{k+1}(r)$ is isomorphic to $\mG^k((0,0))$ and $(0,k)$ has no $V$ successor in $G^k((0,0))$ we know that $f_k(y)\neq a$ for every $y\in G^k((0,0))$. Note that due to $\phi_{\mathrm{edge}}$ and since $f_k((0,k))$ is a west border point and has no $H$-predecessors in $\mA$, $a$ is also a west border point and has no $H$-predecessor in $\mA$. Thus $\mA\upharpoonright (\rng[h]\cup\{a\})$ is in-degree complete. Since $\mA\upharpoonright \rng[h]$ is in-degree complete, $a$ has no $V\cup H$-successors in $\rng[h]$. Due to $\phi_\mathrm{edge}$ and $\phi_\mathrm{injective}$, $a$ has no reflexive loops. We extend $h$ by $h((0,k+1)):=a$. Clearly the extended function $h$ is an isomorphism and $\mA\upharpoonright \rng[h]$ in-degree complete.


Now let $m\in\{0,\dots,k-1\}$ and assume that $h((j,(k+1)-j))$ is defined for all $j \leq m$, $h$ is an isomorphism extending $f_k$ and $h(G)$ is in-degree complete. We will prove that we can extend $h$ by assigning a value for $h(m+1,(k+1)-(m+1))$, still maintaining the required properties. By the induction hypothesis we have defined a value for $h((m,(k+1)-m))$.
Now $h((m,(k+1)-m))$ is the $V^2$-successor of $h((m,(k-1)-m))$. Since $h((m,(k-1)-m))$ has no $H^2$ successor in the structure $\mA\upharpoonright \rng[h]$, the $H^2$- and $V^2$-successors of $h((m,(k-1)-m))$ in $\mA$ cannot be the same point.
Now by Remark~\ref{gridlike description}\ref{join description}, this implies that there is a point $c\in A\setminus \rng[h]$ such that $c$ is the $H\circ V$- and $V\circ H$-successor of $h((m,(k-1)-m))$ in $\mA$.
We extend $h$ by $h((m+1,k-m)):=c$ and observe that $\mA\upharpoonright (\rng[h]\cup \{c\})$ is still in-degree complete. By $\phi_\mathrm{injective}$ and in-degree completeness of $\mA\upharpoonright(\rng[h]\setminus\{c\})$, the extended function $h$ is an isomorphism.

Finally we extend the south border. This is possible by reasoning similar to the case where we extended the west border. 

Let $f_{k+1}$ be the isomorphism from $\mG^{k+1}((0,0))$ to $\mA^{k+1}(r)$ that exists by the inductive proof. Clearly $\mA^{k+1}(r)$ is in-degree complete and $f_k\subseteq f_{k+1}$. Now since the isomorphisms $f_i$ for $i\in\N$ constitute an ascending chain, $\bigcup_{i\in\N}f_i$ is an isomorphism from $\mG$ to a substructure of $\mA$. Therefore $\mA$ has an isomorphic copy of the grid as a substructure.
\end{proof}

The last tool needed to prove the main theorem is the following trivial lemma.

\begin{lemma}\label{tiling supergrids}
Let $T$ be a set of tiles and $\mB=(B,V,H)$ a structure. Then $\mB$ is $T$-tilable iff there is a structure $\mA$ which is $T$-tilable and contains a substructure that is isomorphic to $\mB$.
\end{lemma}

The following is the main theorem of this section.

\begin{theorem}\label{iftwo sat complexity}
\sat[\iftwo] is \pizeroone-complete.
\end{theorem}
\begin{proof}
For the upper bound note that $\sat[\fo]\in\pizeroone$ by G\"odel's completeness theorem. By Remark~\ref{eso sat} it follows that $\sat[\eso]\in\pizeroone$ and by the computable translation from $\df$ into $\eso$ from \cite[Theorem~6.2]{va07}, it follows that $\sat[\df^3] \in \pizeroone$. Finally, the computability of the reductions in Lemma~\ref{iftod} and Theorem ~\ref{dtwo le if2} implies $\sat[\iftwo]\in \pizeroone$.

The lower bound follows by the reduction $g$ from \tiling[\mG] to our problem defined by $g(T):=\phi_{\mathrm{grid}} \wedge \phi_T$.
To see that $g$ indeed is such a reduction, first let $T$ be a set of tiles such that $\mG$ is $T$-tilable. Then, by Lemma~\ref{tiling iff tiling formula}, it follows that there is an expansion $\mG^*$ of $\mG$ such that $\mG^*\models \phi_T$. Clearly, $\mG^*\models \phi_{\mathrm{grid}}$ and therefore $\mG^* \models \phi_{\mathrm{grid}}\wedge \phi_T$.
If, on the other hand, $\mA^*$ is a structure such that $\mA^*\models \phi_{\mathrm{grid}}\wedge \phi_T$, then by Theorem~\ref{gridlike includes grid}, the $\{V,H\}$-reduct $\mA$ of $\mA^*$ contains an isomorphic copy of $\mG$ as a substructure. Furthermore, by Lemma~\ref{tiling iff tiling formula}, $\mA$ is $T$-tilable. Hence, by Lemma~\ref{tiling supergrids}, $\mG$ is $T$-tilable.
\end{proof}

\subsection{\texorpdfstring{Finite satisfiability for \iftwo is undecidable}{Finite satisfiability for IF\textasciicircum2 is undecidable}}\label{iffinsatsection}
We will now discuss the problem \finsat[\iftwo] whose undecidability proof is similar to the above, the main difference being that it uses tilings of tori instead of tilings of the grid.
\begin{definition}\label{toruslike description}

A finite structure $\mA=(A,V,H,V',H')$ is \emph{torus-like} iff it satisfies the following two conditions
\begin{enumerate}[i)]
\item there exist unique and distinct points $SW$, $NW$, $NE$, $SE$ such that
\begin{enumerate}
\item $SW$ has no $V$- and no  $H$-predecessor,
\item $NW$ has no $H$-predecessor and no $V$-successor,
\item $NE$ has no $V$- and no $H$-successor and
\item $SE$ has no $H$-successor and no $V$-predecessor,
\end{enumerate}
\item there exist $m,n\in\N$ such that
\begin{enumerate}
\item $(A,V,H)$ is a model that has an isomorphic copy of the $m\times n$ grid as a component with $SW$, $NW$, $NE$ and $SE$ as corner points,
\item $(A,V',H)$ is a model that has an isomorphic copy of the $m\times 2$ grid as a component with $NW$, $SW$, $SE$ and $NE$ as corner points and $(NW,SW),(NE,SE)\in V'$,
\item $(A,V,H')$ is a model that has an isomorphic copy of the $2\times n$ grid as a component with $SE$, $NE$, $NW$ and $SW$ as corner points and $(SE,SW),(NE,NW)\in H'$.
\end{enumerate}
By a \emph{component} of $\mA=(A,V,H)$ we mean a maximal weakly connected substructure $\mM$, i.e., any two points in $M$ are connected by a path along $R:=V\cup H\cup V^{-1}\cup H^{-1}$, and furthermore, for all $M'$ such that $M\subset M' \subseteq A$, there exist two points in $\mA\upharpoonright M'$ that are not connected by $R$.
\end{enumerate}
\end{definition}

In order to define torus-likeness of a structure with an \iftwo formula we first need to express that a finite structure has a finite grid as a component. This is done in essentially the same way as expressing that a structure has a copy of the infinite grid as a substructure.

\begin{definition}\label{def:fingridlike}
A finite structure $\mA=(A,V,H)$ is called \emph{fingrid-like} iff
it satisfies the conjunction $\phi_\mathrm{fingrid}$ of the formulas
\[\begin{array}{l@{\ }c@{\ }l}
\phi_{\mathrm{SWroot}}           & :=    &\exists x\forall y \big(\neg V(y,x)\wedge \neg H(y,x) \wedge \exists y V(x,y)\wedge\exists y H(x,y)\big),\\ 
\phi_{\mathrm{functional}}(R)            & :=    & \forall x\forall y \big(R(y,x)\,\to\, \exists y/\{x\}\,x=y \big)\\
                                &       & \quad\text{for $R\in\{V,H\}$},\\
\phi_{\mathrm{injective}}(R)             & :=    & \forall x\forall y \big(R(x,y)\, \to\, \exists y/\{x\} \,x=y \big)\\
                                &       & \quad\text{for $R\in\{V,H\}$},\\
\phi_{\mathrm{distinct}}                 & :=    & \forall x\forall y\, \neg \big(V(x,y)\wedge H(x,y)\big),\\
\phi_{\mathrm{SWedge}}        & :=    & \forall x \Big(\big(\forall y \,\neg R(y,x)\big)\to\forall y \big((R'(x,y)\vee R'(y,x)) \to \forall x\,\neg R(x,y)\big)\Big)\\
                                &       & \quad\text{for $(R,R')\in\{(V,H),(H,V)\}$,}\\
\phi_{\mathrm{NEedge}}        & :=    & \forall x \Big(\big(\forall y \,\neg R(x,y)\big)\to\forall y \big((R'(x,y)\vee R'(y,x))\to \forall x\,\neg R(y,x)\big)\Big)\\
                                &       & \quad\text{for $(R,R')\in\{(V,H),(H,V)\}$,}\\
\phi_{\mathrm{finjoin}}          & :=    & \forall x \Big(\forall y \neg V(x,y)\vee \forall y \neg H(x,y)\vee \forall y \Big(\big(V(x,y)\vee H(x,y)\big)\\
                                &       & \quad\to \exists x/\{y\}\, \big(V(y,x)\vee H(y,x)\big)\Big)\Big),
\end{array}\]
\end{definition}

Fingrid-likeness can also be described in the following intuitive way.
\begin{remark}\label{fingridlike description}
A structure $\mA=(A,V,H)$ is fingrid-like iff
\begin{enumerate}[i)]
\item $V$ and $H$ are (graphs of) injective partial functions, i.e.,~the in- and out-degree of every element is at most one ($\phi_{\mathrm{functional}}$ and $\phi_{\mathrm{injective}}$),
\item there exists a point, denoted by $SW$, that has a $V$-successor and an $H$-successor but does not have $V\cup H$-predecessors,
($\phi_{\mathrm{SWroot}}$),
\item for every element, its $V$-successor is distinct from its $H$-successor ($\phi_{\mathrm{distinct}}$),
\item for every element $x$ such that $x$ does not have a $V$ (resp.~$H$) predecessor, the $H$ (resp.~$V$) successor and predecessor of $x$ also do not have a $V$ (resp.~$H$) predecessor ($\phi_{\mathrm{SWedge}}$),
\item for every element $x$ such that $x$ does not have a $V$ (resp.~$H$) successor, the $H$ (resp.~$V$) successor and predecessor of $x$ also do not have a $V$ (resp.~$H$) successor ($\phi_{\mathrm{NEedge}}$),
\item for every element $x$  that has a $V$-successor and an $H$-successor there is an element $y$ such that $(x,y) \in (V\circ H)\cap (H\circ V)$ or $(x,y) \in (V\circ V)\cap (H\circ H)$.
\end{enumerate}
\end{remark}

Notice that for a grid $\mG$ to be grid-like, it is required that the grid is not of the type $1\times n$ or $n\times 1$ for any $n\in\N$.
A grid that is grid-like is called a \emph{proper} grid. Now we can show that a fingrid-like structure contains a proper finite grid as a component.
\begin{lemma}\label{fingridlike includes grid}
Let $\mA=(A,V,H)$ be a finite fingrid-like structure. Then $\mA$ contains an isomorphic copy of a proper finite grid as a component.
\end{lemma}
\begin{proof}
%
Due to $\phi_\mathrm{SWroot}$ there exists a point denoted by $SW\in A$ that has a $V$-successor and an $H$-successor, but has no $V\cup H$-predecessors. Now since $V$ is an injective partial function and $A$ is finite, there exists $n\in\N$ such that  for all $x\in A$ $(SW,x)\notin V^{n+1}$. For similar reasons there exists $m\in\N$ such that for all $x\in A$, $(SW,x)\notin H^{m+1}$.
Let $m$ and $n$ be the smallest such numbers. We will show that $\mA$ has an isomorphic copy of the $m\times n$ grid as a component.

We will first show by induction on $k\le n$ that $\mA$ has an isomorphic copy of the $m\times k$ grid as an in-degree complete substructure with $SW$ as a corner point. By the selection of $m$ the point $SW$ has a $H^i$ successor $v_i$ for each $i\leq m$. Since $H$ is an injective partial function and $SW$ has no $H$-predecessors, the points $v_i$ are all distinct and unique. Due to $\phi_\mathrm{SWedge}$ none of the points $v_i$ has a $V$-predecessor and therefore the $V$-successors of the points $v_i$ are not in the set $\{v_i\mid i\leq m\}$. Therefore $\mA\upharpoonright\{v_i\mid i\leq m\}$ is an isomorphic copy of the $m\times 1$ grid. Due to $\phi_\mathrm{SWedge}$, $\phi_\mathrm{SWroot}$ and $\phi_\mathrm{injective}$ the structure $\mA\upharpoonright\{v_i\mid i\leq m\}$ is in-degree complete.

Let us then assume that $\mB$, an in-degree complete substructure of $\mA$, is an isomorphic copy of the $m\times k$ grid $\mG_{(m,k)}$ with $SW$ as a corner point and $k<n$. Let $h$ be the corresponding isomorphism from $\mG_{(m,k)}$ to $\mB$. We will now extend $h$ to $h'$ such that $h'$ is an isomorphism from the $m\times (k+1)$ grid to an in-degree complete substructure of $\mA$. Since $k+1\leq n$ there exists a point $a_0\in A$ such that $a_0$ is the $V$-successor of $h((0,k))$.
Due to $\phi_\mathrm{NEedge}$ and since $h((0,k))$ has a $V$-successor, each of the points $h((i,k))$, $i\leq m$, has a $V$-successor $a_i$. Since $V$ is a partial injective function and the points $h((i,k))$ are all distinct, the points $a_i$ are also all distinct. The structure $\mB$ is in-degree complete, and hence neither any of the points $a_i$ nor any of their $V\cup H$-successors is in $B$.

We will next show that $(a_i,a_{i+1})\in H^\mA$ for all $i<m$. For $i\leq m-2$, the point $h((i,k))$ has an $H^2$-successor but has no $V^2$-successor in the structure $\mB$. Therefore for all $i\leq m-2$, if the $V^2$-successor of $h((i,k))$ exists in $\mA$, it cannot be the same as the $H^2$-successor of $h((i,k))$. Notice that each of the points $h((i,k))$, $i\leq m-2$, has a $V$- and $H$-successor in $\mA$. Therefore due to $\phi_\mathrm{finjoin}$ the $V\circ H$-successor and the $H\circ V$-successor of the point $h((i,k))$, $i\leq m-2$, are the same. Therefore the $H$-successor of $a_i$ is $a_{i+1}$ for all $i\leq m-2$.

It needs still to be shown that $(a_{m-1},a_m)\in H^\mA$.
The point $h((m-1,k))$ has no $H^2$ successor in $\mA$ since $h((m,k))$ is an east border point (due to $\phi_\mathrm{NEedge}$ and the selection of $m$). Therefore there cannot be a point $a$ in $\mA$ such that it is both an $H^2$-successor and a $V^2$-successor of $h((m-1,k))$. Now due to $\phi_\mathrm{finjoin}$ and the fact that $h((m-1,k))$ has a $V$- and an $H$-successor in $\mA$, the $H\circ V$-successor and $V\circ H$-successor of $h((m-1,k))$ have to be the same point. Therefore $(a_{m-1},a_m)\in H^\mA$.

We define $h':=h\cup \{((i,k+1),a_i)\mid i\leq m\}$. Each point $a_i$ with the exception of the west border point $a_0$ has an $H$-predecessor $a_{i-1}$. Hence, due to the in-degree completeness of $\mB$, injectivity of $V$ and $H$, and since each of the points $a_i$ has a $V$-predecessor in the set $B$, we conclude that the structure $\mA\upharpoonright\rng[h']$ is an in-degree complete substructure of $\mA$. We also notice that due to injectivity, the points $a_i$ have no reflexive loops. Due to in-degree completeness of $B$, none of the $V\cup H$-successors of the points $a_i$, $i\leq m$, are in the set $B$. Hence it is sraightforward to observe that $h'$ is the desired isomorphism from the $m\times (k+1)$ grid to an in-degree complete substructure of $\mA$.

We have now proven that $\mA$ has an isomorphic copy of the $m\times n$ grid as a substructure with $SW$ as a corner point. Let $h$ be the isomorphism from the $m\times n$ grid to a substructure of $\mA$ with $SW$ as a corner point. By the selection of $m$ and $n$, the point $h((0,n))$ has no $V$-successors and $h((m,0))$ has no $H$-successors. Therefore, due to $\phi_\mathrm{NEedge}$, none of the points $h((i,n))$, $i\leq m$, have a $V$-successor and none of the points $h((m,j))$, $j\leq n$, have a $H$-successor. This together with functionality and injectivity of $H$ and $V$, and the fact that west border points have no $H$-predecessors and south border points have no $V$-predecessors, imply that $\mA$ has an isomorphic copy of the $m\times n$ grid as a component. Since the point $SW$ has a $V$-successor and an $H$-successor, the $m\times n$ grid is a proper grid.
\end{proof}

We now define some auxiliary \fotwo-formulas.

\[\begin{array}{lcl} 
\phi_{\mathrm{NStape}}           & :=    & \exists x(\phi_{\mathrm{SW}}^{(V,H)}(x)\wedge\phi_{\mathrm{NW}}^{(V',H)}(x))\wedge\exists x(\phi_{\mathrm{SE}}^{(V,H)}(x)\wedge\phi_{\mathrm{NE}}^{(V',H)}(x))\\
																&				&	\wedge\exists x(\phi_{\mathrm{NW}}^{(V,H)}(x)\wedge\phi_{\mathrm{SW}}^{(V',H)}(x))\wedge\exists x(\phi_{\mathrm{NE}}^{(V,H)}(x)\wedge\phi_{\mathrm{SE}}^{(V',H)}(x))\\
																&				&	\wedge\exists x\exists y (\phi_{\mathrm{NW}}^{(V,H)}(x)\wedge\phi_{\mathrm{SW}}^{(V,H)}(y)\wedge V'(x,y)),\\\\
																
\phi_{\mathrm{EWtape}}           & :=    & \exists x(\phi_{\mathrm{SW}}^{(V,H)}(x)\wedge\phi_{\mathrm{SE}}^{(V,H')}(x))\wedge\exists x(\phi_{\mathrm{SE}}^{(V,H)}(x)\wedge\phi_{\mathrm{SW}}^{(V,H')}(x))\\
																&				&	\wedge\exists x(\phi_{\mathrm{NW}}^{(V,H)}(x)\wedge\phi_{\mathrm{NE}}^{(V,H')}(x))\wedge\exists x(\phi_{\mathrm{NE}}^{(V,H)}(x)\wedge\phi_{\mathrm{NW}}^{(V,H')}(x))\\
																&				&	\wedge\exists x\exists y (\phi_{\mathrm{SE}}^{(V,H)}(x)\wedge\phi_{\mathrm{SW}}^{(V,H)}(y)\wedge H'(x,y)),\\\\															

\phi_\mathrm{uniquecorners}						& := & \bigwedge\limits_{P\in C}\forall x\forall y ((P(x)\wedge P(y))\to x=y),
\end{array}\]

where $C=\{\phi_{\mathrm{T}}^{(R,S)}(x)\mid T\in\{\mathrm{SW},\mathrm{NW},\mathrm{NE},\mathrm{SE}\}, (R,S)\in \{ (V,H), (V',H), (V,H')\} \}$

and
\[\begin{array}{lcl} 
\phi_{\mathrm{SW}}^{(R,S)}(x)   & :=    & \forall y \big(\neg R(y,x)\wedge \neg S(y,x)\big)\wedge \exists y R(x,y)\wedge\exists y S(x,y),\\
\phi_{\mathrm{NW}}^{(R,S)}(x)               & :=    & \forall y \big(\neg R(x,y)\wedge \neg S(y,x)\big)\wedge \exists y R(y,x)\wedge\exists y S(x,y),\\
\phi_{\mathrm{NE}}^{(R,S)}(x)               & :=    & \forall y \big(\neg R(x,y)\wedge \neg S(x,y)\big)\wedge \exists y R(y,x)\wedge\exists y S(y,x),\\
\phi_{\mathrm{SE}}^{(R,S)}(x)               & :=    & \forall y \big(\neg R(y,x)\wedge \neg S(x,y)\big)\wedge \exists y R(x,y)\wedge\exists y S(y,x),
\end{array}\]
$(R,S)\in \{ (V,H), (V',H), (V,H')\}$.

Let $\mA=(A,V,H,V',H')$ be a finite structure such that the underlying structures $(A,V,H)$, $(A,V',H)$ and $(A,V,H')$ are fingrid-like. In this context the intuitive meaning of the above three formulas is the following.
\begin{itemize}
\item The formula $\phi_\mathrm{uniquecorners}$ expresses that the structures $(A,V,H)$, $(A,V',H)$ and $(A,V,H')$ each have four unique corner points, exactly one of each type, i.e., southwest corner, northwest corner, northeast corner and southeast corner. In each structure the corner points definine a boundary of a proper finite grid.
\item The formula $\phi_{\mathrm{NStape}}$ expresses that the proper finite grid in $(A,V',H)$ is of the type $m\times 2$ and connects the north border of the grid in $(A,V,H)$ to the south border of the grid in $(A,V,H)$. (The grids in $(A,V,H)$ and $(A,V',H)$ form a tube.)
\item  The formula $\phi_{\mathrm{EWtape}}$ expresses that the proper finite grid in $(A,V,H')$ is of the type $2\times n$ and connects the east border of the grid in $(A,V,H)$ to the west border of the grid in $(A,V,H)$. (The grids in $(A,V,H)$ and $(A,V,H')$ form a tube. The three grids together form a torus.)
\end{itemize}
\begin{lemma}\label{lemma:torus}
Let $\mA=(A,V,H,V',H')$ be a finite structure such that the underlying structures $(A,V,H)$, $(A,V',H)$ and $(A,V,H')$ are fingrid-like and the structure $\mA$ satisfies the conjunction of the formulas $\phi_{\mathrm{NStape}}$, $\phi_{\mathrm{EWtape}}$ and  $\phi_\mathrm{uniquecorners}$.
Then \mA is torus-like.
\end{lemma}

Notice that for a torus $\mD$ to be torus-like, it is required that the finite grid $(D,V,H)$ is not of the type $1\times n$ or $n\times 1$ for any $n\in\N$.
A torus that is torus-like is called a \emph{proper} torus.

It immediately follows from the previous lemma that there is a sentence $\phi_{\mathrm{torus}}\in \iftwo$ such that for all finite structures $\mA=(A,V,H,V',H')$, if $\mA \models \phi_{\mathrm{torus}}$ then $\mA$ is torus-like, and furthermore, every proper torus satisfies $\phi_\mathrm{torus}$.

We say that a structure $\mA=(A,\{R_i^{\mA}\}_{i\le n})$ is a \emph{topping} of a structure $\mB=(B,\{R_i^{\mB}\}_{i\le n})$ iff $A=B$ and $R_i^{\mB}\subseteq R_i^{\mA}$ for all $i\le n$.

\begin{lemma}\label{toruslike includes torus}
Let $\mA=(A,V,H,V',H')$ be a finite structure with $\mA\models\phi_{\mathrm{torus}}$. Then there is a torus $\mD$ such that $\mA$ contains an isomorphic copy of a topping of $\mD$ as a substructure.
\end{lemma}
\begin{proof}
Immediate from Definition~\ref{toruslike description} and the definition of a torus, i.e., Definition~\ref{def:grid}.
\end{proof}

The following theorem is the finite analogue of Theorem~\ref{iftwo sat complexity}.

\begin{theorem}\label{iftwo finsat complexity}
\finsat[\iftwo] is \sigmazeroone-complete.
\end{theorem}
\begin{proof}
For the upper bound, note that since all finite structures can be recursively enumerated and the model checking problem of $\iftwo$ over finite models is clearly decidable, we have $\finsat[\iftwo]\in\sigmazeroone$.

The lower bound follows by a reduction $g$ from \tiling[\torus] to our problem defined by $g(T):=\phi_{\mathrm{torus}} \wedge \gamma_T$.
To see that $g$ indeed is such a reduction, first let $T$ be a set of tiles such that there is a torus $\mD'$ which is $T$-tilable. 
Therefore there clearly exists a proper torus $\mD$ that is $T$-tilable. Then, by Lemma~\ref{torustiling iff tiling formula}, it follows that there is an expansion $\mD^*$ of $\mD$ such that $\mD^*\models \gamma_T$. We have $\mD^*\models \phi_{\mathrm{torus}}$ and therefore $\mD^* \models \phi_{\mathrm{torus}}\wedge \gamma_T$.
If, on the other hand, $\mA^*$ is a finite structure such that $\mA^*\models \phi_{\mathrm{torus}}\wedge \gamma_T$, then by Lemma~\ref{toruslike includes torus}, $\mA^*$ has a substructure $\mB^*_+$, which is an expansion of an isomorphic copy of a topping of a torus $\mB$. 
Furthermore, by Lemma~\ref{torustiling iff tiling formula}, the $\{V,V',H,H'\}$-reduct $\mA$ of the structure $\mA^*$ is $T$-tilable. 
Hence, by the obvious analogue of Lemma~\ref{tiling supergrids}, the $\{V,V',H,H'\}$-reduct $\mB_+$ of $\mB^*_+$ is $T$-tilable. Therefore $\mB$ is clearly $T$-tilable.
\end{proof}

\section{\texorpdfstring{Satisfiability for \dtwo is NEXPTIME-complete}{Satisfiability for D\textasciicircum2 is NEXPTIME-complete}}\label{sec:satd2}
In this section we show that \sat[\dtwo] and \finsat[\dtwo] are \NEXPTIME-complete. Our proof uses the fact that \sat[\foctwo] and \finsat[\foctwo] are \NEXPTIME-complete \cite{Pratt-Hartmann:2005}.

\begin{theorem}\label{DtoESO}
Let $\tau$ be a relational vocabulary. For every formula $\phi\in \dtwo[\tau]$ there is a sentence $\phi^*\in \eso[\tau\cup\{R\}]$ (with $\ar[R]=|\fr[\phi]|$),
\[\phi^* :=\exists R_1\ldots\exists R_k\psi,\]
where $R_i$ is of arity at most $2$ and $\psi\in \foctwo$,
 such that for all $\mA$ and teams $X$ with $\dom(X)=\fr(\phi)$ it holds that 
\begin{equation}\label{dtoesoequiv}
\mA\models_X \phi \text{ iff } (\mA,\rel(X))\models \phi^*,
\end{equation}
where $(\mA,\rel(X))$ is the expansion $\mA'$ of $\mA$ into vocabulary $\tau\cup\{R\}$ defined by $R^{\mA'} := \rel[X]$.
\end{theorem}
\begin{proof}
Using induction on $\phi$ we will first translate $\phi$ into a sentence  $\tau_\phi \in \eso[\tau\cup\{R\}]$ satisfying \eqref{dtoesoequiv}. Then we note that $\tau_{\phi}$ can be translated into an equivalent sentence $\phi^*$ that also satisfies the syntactic requirement of the theorem. The proof is a modification of the proof from \cite[Theorem~6.2]{va07}. Below we write $\phi(x,y)$ to indicate that $\fr[\phi]=\{x,y\}$. Also, the quantified relations  $S$ and $T$ below are assumed not to appear in $\tau_{\psi}$ and $\tau_{\theta}$. 
 \begin{enumerate}
\item Let $\phi(x,y)\in \{x=y,\neg x=y, P(x,y),\neg P(x,y)\}$. Then  $\tau_{\phi}$ is defined as
\[ \forall x\forall y( R(x,y)\rightarrow \phi(x,y)).  \]
\item Let $\phi(x,y)$ be of the form $\dep(x,y)$. Then $\tau_{\phi}$ is  defined as
\[\forall x\exists ^{\le 1}yR(x,y).\]
\item Let $\phi(x,y)$ be of the form $\neg \dep(x,y)$. Then $\tau_{\phi}$ is  defined as
\[\forall x\forall y\neg R(x,y).\]
\item Let $\phi(x,y)$ be of the form $\psi(x,y)\vee \theta(y)$. Then $\tau_{\phi}$ is defined as
\[
\exists S\exists T(\tau_{\psi}(R/S)\wedge\tau_{\theta}(R/T) \wedge \forall x\forall y (R(x,y)\rightarrow S(x,y)\vee T(y))).
\]
\item\label{0-ary} Let $\phi(x)$ be of the form $\psi(x)\vee \theta$. Then $\tau_{\phi}$ is defined as
\[
\exists S\exists T(\tau_{\psi}(R/S)\wedge\tau_{\theta}(R/T) \wedge \forall x (R(x)\rightarrow S(x)\vee T)).
\]
\item Let $\phi(x)$ be of the form $\psi(x)\wedge \theta(y)$. Then $\tau_{\phi}$ is defined as
\[
\exists S\exists T(\tau_{\psi}(R/S)\wedge\tau_{\theta}(R/T)\wedge \forall x\forall y (R(x,y)\rightarrow S(x)\wedge T(y))).
\]
\item Let $\phi(x)$ be of the form $\exists y\psi(x,y)$.  Then $\tau_{\phi}$ is defined as 
\[\exists S(\tau_{\psi}(R/S)\wedge \forall x\exists y(R(x)\rightarrow  S(x,y))). \]      
 \item Let $\phi(x)$ be of the form $\forall y\psi(x,y)$. Then  $\tau_{\phi}$ is defined as
\[\exists S(\tau_{\psi}(R/S)\wedge \forall x \forall y(R(x)\rightarrow  S(x,y))).\]
\end{enumerate}
It is worth noting that in the translation above we have not displayed all the possible cases, e.g.,  $\phi$ of the form $\dep[x]$ or $P(x)$, for which $\tau_{\phi}$ is defined analogously to the above.  
Note also that, for convenience,  we allow $0$-ary relations in the translation. The possible interpretations of a $0$-ary relation $R$ are $\emptyset$ and $\{\emptyset\}$. Furthermore, for a $0$-ary $R$, we define $\mA\models R$ if and only if $R^{\mA}=\{\emptyset\}$. Clause \ref{0-ary} exemplifies the use of $0$-ary relations in the translation. It is easy to see that $\tau_{\phi}$ in \ref{0-ary} is equivalent to 
\[
  \exists S(\tau_{\theta}(R/\top)\vee  (\tau_{\psi}(R/S)\wedge \forall x (R(x)\rightarrow S(x)))).
\]
Furthermore, the use of $0$-ary relations in the above translation can be easily eliminated with no essential change in the translation.   

A straightforward induction on $\phi$ shows that $\tau_{\phi}$ can be transformed into $\phi^*$ of the form
\[\exists R_1\ldots\exists R_k (\forall x\forall y\psi\wedge \bigwedge_{i}\forall x \exists y\theta_i\wedge \bigwedge_{j} \forall x\exists y^{\le1}R_{m_j}(x,y)), \]
where $\psi$ and $\theta_i$ are quantifier-free.
\end{proof}
Note that if $\phi\in \dtwo$ is a sentence, the relation symbol $R$ is 0-ary and $\rel[X]$ (and $R^{\mA}$) is either $\emptyset$  or $\{\emptyset\}$.
Hence, Theorem~\ref{DtoESO} implies that for an arbitrary sentence $\phi \in \dtwo[\tau]$ there is a sentence $\phi^*(R/ \true) \in \eso[\tau]$ such that for all $\mA$ it holds that
\begin{equation}\label{d2eso sentences}
\mA \models \phi \text{ iff }\mA \models_{\{\emptyset\}} \phi \text{ iff }\mA \models \phi^*(R/\true).
\end{equation}

It is worth noting that, if $\phi\in\dtwo$ does not contain any dependence atoms, i.e.,~$\phi\in\fotwo$, the sentence $\phi^*$ is of the form
\[\exists R_1\ldots\exists R_k (\forall x\forall y\psi\wedge \bigwedge _{i}\forall x \exists y\theta_i)\]
and the first-order part of this is in Scott normal form. So, in Theorem \ref{DtoESO} we essentially translate formulas of $\dtwo$ into Scott normal form \cite{sc62}.

Theorem \ref{DtoESO} now implies the following:
\begin{theorem}\label{dtwo nexptime}
\sat[\dtwo] and \finsat[\dtwo] are \NEXPTIME-complete.
\end{theorem}
\begin{proof}
Let $\phi\in \dtwo$ be a sentence. Then, by \eqref{d2eso sentences}, $\phi$ is (finitely) satisfiable if and only if $\phi^*$ is. Now $\phi^*$ is of the form
\[\exists R_1\ldots\exists R_k\psi,\]
where $\psi\in \foctwo$. Clearly, $\phi^*$ is (finitely) satisfiable iff
$\psi$ is (finitely) satisfiable as a $\foctwo[\tau\cup\{R_1,\dots,R_k\}]$ sentence. Now since the mapping $\phi\mapsto \phi^*$ is clearly computable in polynomial time and (finite) satisfiability of $\psi$ can be checked in \NEXPTIME \cite{Pratt-Hartmann:2005}, we get that $\sat[\dtwo], \finsat[\dtwo]\in \NEXPTIME$. On the other hand, since $\fotwo\le \dtwo $ and \sat[\fotwo], \finsat[\fotwo] are \NEXPTIME-hard \cite{grkova97}, it follows that \sat[\dtwo] and \finsat[\dtwo] are as well.
\end{proof}

\section{Conclusion}

We have studied the complexity of the two-variable fragments of dependence logic and independence-friendly logic.
We have shown (Theorem ~\ref{dtwo nexptime}) that both the satisfiablity and finite satisfiability problems for \dtwo are decidable, \NEXPTIME-complete to be exact.
We have also proved (Theorems \ref{iftwo sat complexity} and \ref{iftwo finsat complexity}) that both problems are undecidable for \iftwo; the satisfiability and finite satisfiabity problems for \iftwo are \pizeroone-complete and \sigmazeroone-complete, respectively. While the full logics \df and \ifl are equivalent over sentences, we have shown that the finite variable variants \dtwo and \iftwo are not, the latter being more expressive.
This was obtained as a by-product of the deeper result concerning the decidability barrier between these two logics.

There are many open questions related to these logics. We conclude with two of them:
\begin{enumerate}
\item What is the complexity of the validity problems of \dtwo and \iftwo?
\item Is it possible to define NP-complete problems in \dtwo or in \iftwo?
\end{enumerate}
    
\section*{Acknowledgments}
The authors would like to thank  Phokion G.~Kolaitis for suggesting to study the satisfiability of \dtwo. The authors would also like to thank Johannes Ebbing, Lauri Hella, Allen Mann, Jouko V\"a\"an\"anen and Thomas Zeume for helpful discussions and comments during the preparation of this article.


\providecommand{\bysame}{\leavevmode\hbox to3em{\hrulefill}\thinspace}
\providecommand{\MR}{\relax\ifhmode\unskip\space\fi MR }
\providecommand{\MRhref}[2]{%
  \href{http://www.ams.org/mathscinet-getitem?mr=#1}{#2}
}
\providecommand{\href}[2]{#2}

\addcontentsline{toc}{section}{References}

\end{document}